\documentclass[twocolumn,review]{autart}

\usepackage{graphicx}
\usepackage{amsmath}
\usepackage{amssymb}
\usepackage{algorithm,algorithmic}
\usepackage[center]{subfigure}
\usepackage{euscript}
\usepackage{verbatim}
\usepackage{amssymb}
\usepackage{cite}
\usepackage{color}
\usepackage{cancel}
\usepackage{calligra}
\usepackage{pgf}
\usepackage{tikz}
\usepackage{ctable}
\usetikzlibrary{arrows,automata}
\usetikzlibrary{shapes}
\usetikzlibrary{positioning}
\usepackage{bbm}
\usepackage{multirow}
\usepackage{url}

{\normalfont}{\normalfont}
{\normalfont}{\normalfont}
\newtheorem{remark}{Remark}{\normalfont}{\normalfont}
\newtheorem{theorem}{Theorem}
\newtheorem{assumption}{Assumption}

\newtheorem{lemma}{Lemma}
\newenvironment{proof}{\emph{Proof:}}{\hfill$\square$}

\renewcommand{\figurename}{Figure}
\renewcommand{\tablename}{Table}

\usepackage{enumitem}
\renewcommand{\theenumi}{\arabic{enumi}}

\renewcommand{\theenumii}{\arabic{enumii}}

\renewcommand{\theenumiii}{\arabic{enumiii}}

\begin{document}

	\begin{frontmatter}
		\runtitle{Direct design of explicit predictive controllers with priors}
		\title{Data-driven design of explicit predictive controllers\\using model-based priors \thanksref{funding}
		}
		
		\thanks[funding]{This project was partially supported by the Italian Ministry of University and Research under the PRIN'17 project \textquotedblleft Data-driven learning of constrained control systems", contract no. 2017J89ARP.}
		
		\author[Polimi]{Valentina Breschi}\ead{valentina.breschi@polimi.it},
		\author[Polimi]{Andrea Sassella}\ead{andrea.sassella@polimi.it},
		\author[Polimi]{Simone Formentin}\ead{simone.formentin@polimi.it}
		
		\address[Polimi]{Dipartimento di Elettronica, Informazione e Bioingegneria, Politecnico di Milano, Piazza L. da Vinci 32, 20133 Milano, Italy.}
		

		\begin{abstract}	
			In this paper, we propose a data-driven approach to derive explicit predictive control laws, without requiring any intermediate identification step. The keystone of the presented strategy is the exploitation of available priors on the control law, coming from model-based analysis. Specifically, by leveraging on the knowledge that the optimal predictive controller is expressed as a piecewise affine (PWA) law, we directly optimize the parameters of such an analytical controller from data, instead of running an on-line optimization problem. As the proposed method allows us to automatically retrieve also a model of the closed-loop system, we show that we can apply model-based techniques to perform a stability check prior to the controller deployment. The effectiveness of the proposed strategy is assessed on two benchmark simulation examples, through which we also discuss the use of regularization and its combination with averaging techniques to handle the presence of noise. 
		\end{abstract}
		
	\end{frontmatter}
	
	\section{Introduction} 
	For its capability of handling constraints, \emph{Model predictive control} (MPC) is a widely employed technique for advanced control applications (see, \emph{e.g.,} \cite{Rawlings2000,QIN2003,Hrovat2012,Faulwasser2021,Matne2014}). Due to the increasing complexity of the systems to be controlled, the model required by MPC is often no longer parameterized from the physics but learned from data in a \textit{black-box} fashion. However, this identification step generally takes the lion's share of the time and effort required for control design. In the last years, research has thus benched into two main direction. On the one side, efforts has been focused on improving and easing learning procedures \cite{Chiuso2019,LJUNG2010}. On the other side, many approaches have been proposed to directly employ data for the design of predictive controllers, while bypassing any model identification step. Among existing works in this direction, we recall the foundational contributions of \cite{coulson2019data,Berberich2021}, which have been extended to handle tracking problems \cite{Berberich2020b}, to deal with nonlinear systems \cite{Berberich2021b}, and to improve the performance in the presence of noise \cite{Dorfler2021,Coulson2019}, just to mention a few. For both traditional and data-driven predictive controls, the computational effort required to solve the related constrained optimization problem is known to be a potential limit for its application, especially for fast-sampling systems. Nonetheless, when the MPC problem is relatively small, \emph{i.e.,} one has to control a low order system and/or the prediction horizon is relatively short, this limitation can be overcome by explicitly deriving its solution \cite{Bemporad2002b}. In fact, when the cost of the optimization problem is quadratic and the constraints are linear, the explicit solution of MPC is known to be a \emph{Piece-Wise Affine} (PWA) state feedback law. 
	
	In this work, we propose an approach \textit{to directly learn explicit predictive control laws from data}
	. More specifically, we initially build on the foundational results in \cite{DePersis2019,Willems2005} and on model-based priors (namely, the aforementioned fact that the optimal solution of linear MPC is PWA \cite{Bemporad2002b}), to construct a data-driven, multi-step closed-loop predictor. The latter is exploited to construct a fully data-based optimization problem, yielding the estimation of the \textit{parameters} of the optimal predictive control law. As a by-product, we obtain a data-driven description of the closed-loop behavior, which can be used in combination with existing model-based techniques (see, \emph{e.g.,} \cite{Mignone}) to check the stability of the system controlled with the data-driven explicit controller, prior to its deployment. As far as we are aware, this is the first time a data-driven predictive control strategy is provided together with a preliminary assessment of its closed-loop performance.
	
	We should mention here that an early attempt to obtain a data-driven counterpart of explicit MPC was already carried out in \cite{sassella2021learning}. However, the strategy proposed therein relies on an implicit open-loop identification step
	. Another method was presented in \cite{Breschi2021b}, by relying on the behavioral predictor used in \cite{Coulson2019,Berberich2021}. Even though that method involves no open-loop identification phase and it does not require the state to be fully measurable, the latter does not leverage on priors coming from model-based analysis and design. Nonetheless, by leveraging on priors, we are here able to retrieve a data-based characterization of the closed-loop and, thus, practically assess its stability in a data-driven fashion prior to its deployment. This is instead not possible in \cite{Breschi2021b}, where stability cannot be directly assessed nor it is theoretically guaranteed in presence of noisy data. 
	
	The remainder of the paper is organized as follows. The targeted problem is formalized in Section~\ref{sec:problem}, while all the steps required to obtain its data-driven formulation from priors are introduced in Section~\ref{sec:priors}. The explicit control law is derived in Section~\ref{sec:learning}, where we additional discuss practical aspects for its implementation, certified deployment and noise handling. The effectiveness of the proposed strategy is assessed on two benchmark examples in Section~\ref{sec:examples}. Section~\ref{sec:conclusions} concludes the work and indicates some directions for future research.      
	
	\paragraph*{Notation}
	We denote with $\mathbb{N}_{0}$ the set of natural numbers, that includes zero. Let $\mathbb{R}$, $\mathbb{R}^{n}$ and $\mathbb{R}^{n \times m}$ be the set of real numbers, column vectors of dimension $n$ and $n \times m$ dimensional real matrices, respectively. Given $B \in \mathbb{R}^{m \times n}$, its transpose is $B^{\top}$, its Moore-Penrose inverse is $B^{\dagger}$ and, when $m=n$ its inverse is indicated as $B^{-1}$. Given a vector $v \in \mathbb{R}^{n}$, $[v]_{i:j}$ indicates its rows from $i$ to $j$, with $i \leq j \leq n$. For a matrix $B \in \mathbb{R}^{n \times m}$, $[B]_{1:i,1:j}$ denotes a sub-matrix comprising the first $i$ rows and $j$ columns of $B$, with $i \leq n$ and $j \leq m$. Identity matrices are denotes as $I$, while zero matrices and vectors will be denoted as $0$. If a matrix $Q \in \mathbb{R}^{n \times n}$ is positive definite (positive semi-definite), this is denoted as $Q \succ 0$ ($Q \succeq 0$). Given a vector $x \in \mathbb{R}^{n}$, the quadratic form $x'Qx$ is compactly indicated as $\|x\|_{Q}^{2}$. Given a signal $\{\nu_{t} \in \mathbb{R}^{m}\}_{t \in \mathbb{N}_{0}}$ and $0<L<T$, we denote with $N_{0,L,T-1} \in \mathbb{R}^{mL \times T-L+1}$ the associated Hankel matrix
	\begin{equation}
		N_{0,L,T-1}=\begin{bmatrix}
			\nu(0) & \nu(1) & \cdots & \nu(T-L)\\
			\nu(1) & \nu(2) & \cdots & \nu(T-L+1)\\
			\vdots & \vdots & \ddots & \vdots\\
			\nu(L) & \nu(L+1) & \cdots & \nu(T-1) 
		\end{bmatrix},
	\end{equation}
	while, for $i,j \in \mathbb{N}_{0}$, we introduce
	\begin{equation}
		N_{i,T-j}=\begin{bmatrix}
			\nu(i) & \nu(i+1) & \cdots & \nu(T-j)
		\end{bmatrix},~~i<T-j.
	\end{equation}
	
	\section{Problem formulation}\label{sec:problem}
	Consider the class $\mathcal{S}$ of discrete-time \emph{linear, time invariant} (LTI), \emph{controllable} systems with \emph{fully measurable} state, \emph{i.e.,}
	\begin{equation}\label{eq:system}
		\mathcal{S}: \begin{cases}
			x(t+1)=Ax(t)+Bu(t),\\
			y^{\mathrm{o}}(t)=x(t),
		\end{cases}
	\end{equation}
	where $x(t) \in \mathbb{R}^{n}$ denotes the state of $\mathcal{S}$ at time $t \in \mathbb{N}_{0}$, $u(t) \in \mathbb{R}^{m}$ is an exogenous input and $y^{\mathrm{o}}(t) \in \mathbb{R}^{n}$ is the associated \emph{noiseless} output. Let us consider the following \emph{model predictive control} (MPC) problem:
	\begin{subequations}\label{eq:MPC}
		\begin{align}
			& \underset{\{\tilde{u}(k)\}_{k=0}^{L-1}}{\mbox{minimize}}~~~ \sum_{k=0}^{L-1}\left[\|\tilde{x}(k)\|_{Q}^{2}\!+\!\|\tilde{u}(k)\|_{R}^{2}\right]\!+\!\|\tilde{x}(L)\|_{P}^{2} \label{eq:MPCcost}\\
			&\quad \mbox{s.t.}~~ \tilde{x}(k\!+\!1)\!=\!A\tilde{x}(k)\!+\!B\tilde{u}(k),~~k\!=\!0,\ldots,L\!-\!1, \label{eq:MPCmodel}\\
			& \qquad \quad C_{x}\tilde{x}(k)+C_{u}\tilde{u}(k) \leq d,~~k\!=\!0,\ldots,L\!-\!1,\label{eq:MPCconstr}\\
			& \qquad \quad \tilde{x}(0)=x(t). \label{eq:MPCinit}
		\end{align}
	\end{subequations}
	The objective of \eqref{eq:MPC} is to steer both the (predicted) state $\tilde{x}(k)$ and the input $\tilde{u}(k)$ to zero, over a prediction horizon of prefixed length $L>0$. Optimality is indeed dictated by: $(i)$ the distance of the predicted state from zero, penalized with $Q \succeq 0$ over the whole horizon except for the terminal state (weighted via $P \succeq 0$), and $(ii)$ the control effort, penalized via $R \succ 0$. Meanwhile, a set of $n_c$ \emph{polyhedral constraints} dictated by \eqref{eq:MPCconstr} has to be satisfied, with $C_{x} \in \mathbb{R}^{n_{c} \times n}$, $C_{u} \in \mathbb{R}^{n_c \times m}$ and $d \in \mathbb{R}^{n_c}$, while relying on the latest information on the system (see \eqref{eq:MPCinit}).
	Assume also that the matrices $A \in \mathbb{R}^{n \times n}$ and $B \in \mathbb{R}^{n \times m}$ characterizing the dynamics of $\mathcal{S}$ are \emph{unknown}, and that we can access a set of input/output data pairs $\mathcal{D}_{T}=\{\mathcal{U}_{T},\mathcal{Y}_{T}\}$, where $\mathcal{U}_{T}$ and $\mathcal{Y}_{T}$ denote the available input and output sequences, respectively, satisfying the following assumptions.
	\begin{assumption}[Persistently exciting inputs]\label{ass:persistency of excitation}
		The input sequence $\mathcal{U}_{T}=\{u(t)\}_{t=0}^{T}$ is persistently exciting of order $n+L$.
	\end{assumption}
	\begin{assumption}[Noisy outputs]\label{ass:noisy_outputs}
		The output sequence $\mathcal{Y}_{T}=\{y(t)\}_{t=0}^{T}$ is corrupted by noise, namely 
		\begin{equation}\label{eq:noisy_output}
			y(t)=y^{\mathrm{o}}(t)+w(t),
		\end{equation}
		where $v(t)$ is the realization of a zero mean white noise with covariance $\Omega \in \mathbb{R}^{n \times n}$.
	\end{assumption}	 
	\begin{assumption}[Sufficiently long dataset]\label{ass:long_seq}
		The length $T$ of the dataset $\mathcal{D}_{T}$ satisfies the following:
		\begin{equation*}
			T \geq (m+1)(n+L)-1.
		\end{equation*} 
	\end{assumption}		
	The goal of this work is to directly exploit the available data to find an \emph{explicit} solution to \eqref{eq:MPC}, under Assumptions \ref{ass:persistency of excitation}-\ref{ass:long_seq}, while \emph{bypassing} any identification step. 
	
	\section{Exploiting priors for explicit DDPC}\label{sec:priors}
	To attain our goal, it is fundamental to replace the model in \eqref{eq:MPCmodel} with an expression that directly depends on the available data, by also not forgetting what we have learned from model-based predictive control. 
	To start with, we thus recall the following lemma, explicitly stating the form of the optimal explicit predictive controller in our setting \cite{Bemporad2002b}.
	\begin{lemma}[On the solution of \eqref{eq:MPC}]\label{lemma_sol}
		Let $U$ denote the vector stacking the inputs over the prediction horizon, \emph{i.e.,}
		\begin{equation}\label{eq:input_seq}
			U=U(x(t))=\begin{bmatrix}
				\tilde{u}(0)\\
				\tilde{u}(1)\\
				\vdots\\
				\tilde{u}(L-1)
			\end{bmatrix} \in \mathbb{R}^{mL}.
		\end{equation}
		The optimal control sequence $U^{\star}=U^{\star}(x(t))$ solving \eqref{eq:MPC} is a continuous \emph{piecewise affine} (PWA) function of $x(t)$.  
	\end{lemma}
	
	Then, we generalize the one-step ahead predictor introduced in \cite{DePersis2019} to the multi-step case. To this end, it is important to recall that, thanks to Assumptions~\ref{ass:persistency of excitation}, \ref{ass:long_seq} and the \emph{Fundamental Lemma} \cite{Willems2005}, the following rank condition holds:
	\begin{equation}\label{eq:rank}
		\mbox{rank}\left(\begin{bmatrix}
			X_{0,T-L-1}\\ 
			\hline
			U_{0,L,T-L-1}
		\end{bmatrix}\right)=n+mL.
	\end{equation}
	We can now derive the multi-step data-based predictor as follows.
	\begin{theorem}[Data-based multi-step predictor]\label{thm:OL_predictor}
		Let Assumptions~\ref{ass:persistency of excitation} and \ref{ass:long_seq} hold. Given $x(t)$, the sequence of predicted states $\tilde{X}=\tilde{X}(x(t))$ admits the following data-based representation:
		\begin{equation}\label{eq:OLmulti_step}
			\tilde{X}=\begin{bmatrix}
				\tilde{x}(1)\\
				\vdots\\
				\tilde{x}(L)
			\end{bmatrix}=X_{1,L,T-L} \begin{bmatrix}
				X_{0,T-L-1}\\
				\hline
				U_{0,L,T-L-1}
			\end{bmatrix}^{\dagger}\begin{bmatrix}
				x(t)\\
				U
			\end{bmatrix}.
		\end{equation}
	\end{theorem}
	\begin{proof}
		The proof follows the steps of the one in \cite[Appendix B]{DePersis2019}. Specifically, let $v \in \mathbb{R}^{n+mL}$ and $S \in \mathbb{R}^{n+mL \times T}$ be respectively defined as:
		\begin{equation*}
			v:=\begin{bmatrix}
				x(t)\\
				U
			\end{bmatrix},\quad S:=\begin{bmatrix}
				X_{0,T-L-1}\\
				\hline
				U_{0,L,T-L-1}
			\end{bmatrix}, 
		\end{equation*}	
		with $S$ being full row rank, \emph{i.e.,} $\mbox{rank}(S)=n+mL$. By the Rouch\'e-Capelli theorem, for any given $v$, the equality
		\begin{equation*}
			v=S\alpha,
		\end{equation*} 
		admits infinite solutions $\alpha$ of the form
		\begin{equation}\label{eq:alpha}
			\alpha=S^{\dagger} v+ \Pi_{S}^{\perp}w, \forall w \in \mathbb{R}^{T},
		\end{equation}
		with $\Pi_{S}^{\perp}=(I-S^{\dagger}S)$ being the orthogonal projector onto the kernel of $S$. Meanwhile, based on the model in \eqref{eq:MPCmodel}, the predicted state sequence can be defined as a function of $x(t)$ and $U$ as:
		\begin{equation}\label{eq:MBpredictor}
			\tilde{X}=\underbrace{\begin{bmatrix}
					A\\
					\vdots\\
					A^{L\!-\!1}
			\end{bmatrix}}_{\xi}x(t)+\underbrace{\begin{bmatrix}
					B & 0 &  \cdots & 0\\
					AB & B &  \cdots & 0\\
					\vdots & \vdots  &\ddots  & \vdots\\
					A^{L\!-\!1}B & A^{L\!-\!2}B &  \cdots & B
			\end{bmatrix}}_{\Gamma}U.
		\end{equation}
		In turn, such a sequence can be recast as a function of $\alpha$, \emph{i.e.,}
		\begin{equation*}
			\tilde{X}=\begin{bmatrix}
				\xi & \Gamma
			\end{bmatrix}S\alpha=X_{1,L,T-L}\alpha.
		\end{equation*}
		where the second equality straightforwardly follows from \eqref{eq:MBpredictor} and the definition of $S$. By replacing $\alpha$ with \eqref{eq:alpha}, we then obtain
		\begin{equation*}
			\tilde{X}=X_{1,L,T-L}\left(S^{\dagger}v+\Pi_{S}^{\perp}w\right)=X_{1,L,T-L}S^{\dagger}v,
		\end{equation*}
		as $X_{1,L,T-L}\Pi_{S}^{\perp}=\begin{bmatrix} \xi & \Gamma \end{bmatrix}S\Pi_{S}^{\perp}=0$, based on the definition of the projector.
	\end{proof}
	
	This preliminary result allows us to exploit priors on the solution of \eqref{eq:MPC} for the definition of the DDPC problem. In fact, according to Lemma \ref{lemma_sol}, we can parameterize the control sequence $U$ as:
	\begin{equation}\label{eq:param_input}
		U\!=\!\begin{cases}
			K_{1}x(t)\!+\!f_{1},~ \mbox{ if } H_{1}x(t) \leq \ell_{1},\\
			\vdots\\
			K_{M}x(t)\!+\!f_{M},~ \mbox{ if } H_{M}x(t) \leq \ell_{M},
		\end{cases}
	\end{equation}
	where $K_{i} \in \mathbb{R}^{mL \times n}$ and $f_{i} \in \mathbb{R}^{mL}$ are the (unknown) feedback and affine gains characterizing the control law, $\{H_{i},\ell_{i}\}_{i=1}^{M}$ dictates the associated polyhedral partition, for $i=1,\ldots,M$, and the amounts of modes $M$ is dictated by the number of possible combinations of active constraints. Therefore, for a given state $x(t)$, the input sequence is the affine function
	\begin{equation}\label{eq:single_input}
		U(x(t))=Kx(t)+f,
	\end{equation}
	with $K=K_{s(x(t))}$ and $f=f_{s(x(t))}$ denoting the gains associated to the active control law, and with
	\begin{equation*}
		s(x(t))=i \quad \iff \quad H_{i}x(t) \leq \ell_{i}, i \in \{1,\ldots,M\}.
	\end{equation*}
	
	Based on this parameterization, we can compute a data-based closed-loop characterization of the predictor in \eqref{eq:MPCmodel}, as outlined in the following theorem.  
	\begin{theorem}[Closed-loop multi-step predictor]\label{thm:CL_predictor}
		Let Assumptions~\ref{ass:persistency of excitation} and \ref{ass:long_seq} hold. Given $x(t)$, the sequence of predicted states $\tilde{X}=\tilde{X}(x(t))$ can be equivalently expressed as: 
		\begin{subequations}\label{eq:CLpred}
			\begin{equation}
				\tilde{X}=X_{1,L,T-L}\left(G_{K}x(t)+G_{f}\right),
			\end{equation}
			with $G_{K} \in \mathbb{R}^{T-L \times n}$ and $G_{f} \in \mathbb{R}^{T-L}$ satisfying:
			\begin{align}
				&\begin{bmatrix}
					I\\
					K
				\end{bmatrix}=\begin{bmatrix}
					X_{0,T-L-1}\\
					\hline
					U_{0,L,T-L-1}
				\end{bmatrix} G_{K} \label{eq:def_Gk}\\
				& \begin{bmatrix}
					0\\
					f
				\end{bmatrix}=\begin{bmatrix}
					X_{0,T-L-1}\\
					\hline
					U_{0,L,T-L-1}
				\end{bmatrix} G_{f} \label{eq:def_Gf}
			\end{align}
		\end{subequations}
		where $K$ and $f$ characterize the local control law \eqref{eq:single_input}. Accordingly, the input sequence is given by:
		\begin{equation}\label{eq:input_sequence}
			U=U_{0,L,T-L-1}(G_{K}x(t)+G_{f}).
		\end{equation} 
	\end{theorem} 
	\begin{proof}	
		For the Rouch\'e-Capelli theorem, there exists an $T \times n$ matrix $G_{k}$ and a $T$-dimensional vector $G_{f}$ such that \eqref{eq:def_Gk}-\eqref{eq:def_Gf} hold. Meanwhile, replacing \eqref{eq:single_input} into the open-loop predictor in \eqref{eq:OLmulti_step}, the predicted state sequence can be represented as:
		\begin{equation*}
			\tilde{X}=X_{1,L,T-L}\begin{bmatrix}
				X_{0,T-L-1}\\
				\hline
				U_{0,L,T-L-1}
			\end{bmatrix}^{\dagger}\left(\begin{bmatrix}
				I\\
				K
			\end{bmatrix}x(t)+\begin{bmatrix}
				0\\
				f
			\end{bmatrix}\right).
		\end{equation*}
		By combining these two results, the closed-loop representation in \eqref{eq:CLpred} and the equivalent definition of the input sequence in \eqref{eq:input_sequence} straightforwardly follow. 
	\end{proof}
	
	By leveraging on \eqref{eq:CLpred}-\eqref{eq:input_sequence}, we can now equivalently recast the predictive control task in \eqref{eq:MPC} as an optimization problem with the closed-loop matrices $G_{K}, G_{f}$ being the decision variables, as follows\footnote{Notice that the term in the cost that depends only on $x(t)$ has been neglected. This can be done without loss of generality, as the optimal solution does not change.}:
	\begin{subequations}\label{eq:DDPC}
		\begin{align}
			&\underset{G_{K},G_{f}}{\mbox{minimize}}~~~\tilde{X}^{\top}\mathcal{Q}\tilde{X}\!+\!U^{\top}\mathcal{R}U \label{eq:DDPCcost}\\
			&\qquad \mbox{s.t.}~~~ \tilde{X}=X_{1,L,T-L}\left(G_{K}x(t)+G_{f}\right) \label{eq:DDPCmodel1}\\
			& \qquad \qquad~ U=U_{0,L,T-L-1}(G_{K}x(t)+G_{f}), \label{eq:DDPCmodel2}\\
			& \qquad \qquad~ \begin{bmatrix}
				C_{x} & 0\\
				0 & \mathcal{C}_{x}
			\end{bmatrix}\begin{bmatrix}x(t)\\\tilde{X}\end{bmatrix}+\mathcal{C}_{u}U \leq \mathcal{D},\label{eq:DDPCconstr}\\
			& \qquad \qquad~ X_{0,T-L-1}G_{K}x(t)=x(t), \label{eq:DDPCbuild1}\\
			& \qquad \qquad~ X_{0,T-L-1}G_{f}=0, \label{eq:DDPCbuild2}
		\end{align}
	\end{subequations}
	In \eqref{eq:DDPC}, $\mathcal{Q}\!=\!\mbox{diag}\left([Q,\cdots,Q,P]\right)$, $\mathcal{R}\!=\!\mbox{diag}\left([R,\cdots,R]\right)$, $\mathcal{C}_{u}\!=\!\mbox{diag}\left([C_{u},\cdots,C_{u}]\right)$ and
	\begin{align*}
		& \mathcal{C}_{x}\!=\!\!\begin{bmatrix}
			C_{x} & 0 & \cdots & 0 & 0\\
			0 & C_{x} & \cdots & 0 & 0\\
			\vdots & \vdots & \ddots & \vdots & \vdots\\
			0 &  0 & \cdots & C_{x} & 0
		\end{bmatrix}\!\!,~~~~\mathcal{D}=\begin{bmatrix}
			d\\
			\vdots\\
			d
		\end{bmatrix}\!\!.
	\end{align*}
	Note that, the last constraints (see \eqref{eq:DDPCbuild1}-\eqref{eq:DDPCbuild2}) are introduced for the problem to be consistent with the closed-loop representation in \eqref{eq:CLpred}. 
	
	This shift from an open-loop predictor to its closed-loop counterpart allows us to directly learn the control law from data, and avoid any system identification step.	
	
	\begin{remark}\label{remark:noisefree}
		Both the equivalences in \eqref{eq:OLmulti_step} and \eqref{eq:CLpred} exactly hold in a \emph{noiseless} setting only. As such, problem~\eqref{eq:DDPC} is equivalent to \eqref{eq:MPC} only when the available batch of data is noise-free. 
	\end{remark}
	
	\section{Learning  Explicit DDPC} \label{sec:learning}
	To derive its explicit solution, the problem in \eqref{eq:DDPC} is manipulated to obtain a \emph{multi-parametric Quadratic Program} (mp-QP). As a preliminary step, we condense the unknowns of \eqref{eq:DDPC} into a single variable:
	\begin{equation}\label{eq:new_unknown}
		G(t)=\begin{bmatrix}
			G_{K}x(t)\\
			G_{f}
		\end{bmatrix} \in \mathbb{R}^{2(T-L)}.
	\end{equation}
	Accordingly, we can recast \eqref{eq:DDPC} as the following mp-QP
	\begin{subequations}\label{eq:mpQP_ddpc}
		\begin{align}
			&\underset{G(t)}{\mbox{minimize}}~~~(G(t))^{\top}\mathcal{H}_{d}G(t) \label{eq:mpQPcost}\\
			&\qquad \mbox{s.t.}~\quad \Xi_{d} G(t)+\Psi x(t)\leq \mathcal{D}, \label{eq:mpQPconstr}\\
			& \qquad \qquad~~ \Theta_{d} G(t)=\begin{bmatrix}
				x(t)\\0
			\end{bmatrix}, \label{eq:mpQPbuild}
		\end{align}
	\end{subequations}
	where
	\begin{subequations}
		\begin{align}\label{eq:cost_matrix}
			& \mathcal{H}_{d}=\mathcal{X}^{\top}\mathcal{Q}\mathcal{X}+\mathcal{V}^{\top}\mathcal{R}\mathcal{V},\\
			& \Xi_{d}=\begin{bmatrix}
				C_{x} & 0\\
				0 &\mathcal{C}_{x}
			\end{bmatrix}\begin{bmatrix}
				0\\
				\mathcal{X}
			\end{bmatrix}+\mathcal{C}_{u}\mathcal{V},~~\Psi\!=\!\begin{bmatrix}
				C_{x} & 0\\
				0 &\mathcal{C}_{x}
			\end{bmatrix}\begin{bmatrix}
				I\\
				0
			\end{bmatrix}\!\!,\\
			& \Theta_{d}\!=\!\mbox{diag}\left([X_{0,T\!-L\!-1},X_{0,T\!-L\!-1}]\right),,
		\end{align}
		and
		\begin{align*}
			&\mathcal{X}\!=\!\begin{bmatrix}
				X_{1,L,T\!-L} & X_{1,L,T\!-L}
			\end{bmatrix},\\
			& \mathcal{V}\!=\!\begin{bmatrix}U_{0,L,T\!-L\!-1} & U_{0,L,T\!-L\!-1}\end{bmatrix}.
		\end{align*}   
	\end{subequations}
	By focusing on $\mathcal{H}_{d}$ in \eqref{eq:mpQPcost}, it can be proven that this weighting matrix satisfies the following lemma.
	\begin{lemma}[Features of $\mathcal{H}_{d}$]
		Under Assumptions~\ref{ass:persistency of excitation} and \ref{ass:long_seq}, $\mathcal{H}_{d}$ is positive semi-definite.
	\end{lemma}
	\begin{proof}
		This is a direct consequence of Assumptions~\ref{ass:persistency of excitation} and \ref{ass:long_seq}, for which $\mathcal{V}^{\top} \in \mathbb{R}^{2(T-L) \times mL}$ is not full row rank.
	\end{proof}
	
	As the cost should be strictly convex for a unique explicit solution to be retrieved, this feature of $\mathcal{H}_{d}$ prevents us from deriving the explicit law. To overcome this limitation, we introduce a regularization term in the cost of \eqref{eq:mpQP_ddpc}, thus replacing the weight $\mathcal{H}_{d}$ with:
	\begin{equation}\label{eq:reg_cost}
		\mathcal{H}_{d}^{{\gamma}}=\frac{1}{2}\left(\mathcal{H}_{d}+\gamma I\right),
	\end{equation}
	where $\gamma>0$ is an hyper-parameter to be tuned\footnote{The cost has been normalized to ease the subsequent derivations.}. The data-driven control problem then corresponds to the \emph{regularized} mp-QP:
	\begin{subequations}\label{eq:regmpQP_ddpc}
		\begin{align}
			&\underset{G(t)}{\mbox{minimize}}~~~(G(t))^{\top}\mathcal{H}_{d}^{\gamma}G(t) \label{eq:regmpQPcost}\\
			&\qquad \mbox{s.t.}~\quad \Xi_{d} G(t)+\Psi x(t)\leq \mathcal{D}, \label{eq:regmpQPconstr}\\
			& \qquad \qquad~~ \Theta_{d} G(t)=\begin{bmatrix}
				x(t)\\0 \end{bmatrix}. \label{eq:regmpQPbuild}
		\end{align}
	\end{subequations}
	
	\subsection{Derivation of the explicit DDPC law}
	The introduction of the regularizer in \eqref{eq:regmpQP_ddpc} allows us to derive the explicit DDPC law through the manipulation of the  \emph{Karush-Kuhn-Tucker} (KKT) conditions associated with the new DDPC problem. To ease the computations, let us consider the following further assumption. 
	\begin{assumption}[Non-degenerate constraints]\label{ass:nondeg}
		The active constraints of \eqref{eq:regmpQP_ddpc} are linearly independent.
	\end{assumption}   
	Based on this assumption, we now follow the same steps used to derive the explicit model-based predictive control law in \cite{Bemporad2002b}. 
	
	The KKT conditions for the regularized DDPC problem in \eqref{eq:regmpQP_ddpc} are:
	\begin{subequations}\label{eq:KKT}
		\begin{align}
			& \mathcal{H}_{d}^{\gamma} G(t)+\Xi_{d}^{\top}\lambda+\Theta_{d}^{\top}\mu=0, \label{eq:KKT1}\\
			& \lambda^{\top}(\Xi_{d} G(t)+\Psi x(t)-\mathcal{D})=0, \label{eq:KKT2}\\
			& \lambda \geq 0,\label{eq:KKT3}\\
			& \Xi_{d} G(t)+\Psi x(t)-\mathcal{D} \leq 0, \label{eq:KKT4}\\
			& \Theta_{d}G(t)-\begin{bmatrix} x(t)\\ 0
			\end{bmatrix}=0, \label{eq:KKT5}
		\end{align}
		where $\lambda$ and $\mu$ are the Lagrange multipliers associated with inequality and equality constraints in \eqref{eq:regmpQPconstr} and \eqref{eq:regmpQPbuild}, respectively. 
	\end{subequations}
	Let us focus on the $i$-th set of active constraints only, distinguishing between the Lagrange multipliers associated with a given active and inactive inequality constraints. We respectively denote them as $\tilde{\lambda}_{i}$ and $\bar{\lambda}_{i}$. It is straightforward to notice that the combination of \eqref{eq:KKT2} and \eqref{eq:KKT3} leads to the following condition on $\bar{\lambda}_{i}$: 
	\begin{equation*}
		\bar{\lambda}_{i}=0.
	\end{equation*}
	By merging \eqref{eq:KKT2} and \eqref{eq:KKT5} for the $i$-th set of active constraints, it is also straightforward to show that the optimal solution $G_{i}(t)$ satisfies 
	\begin{equation}\label{eq:ActiveCond}
		\Phi_{d,i}G_{i}(t)-
		\tilde{S}_{i}x(t)-\tilde{W}_{i}=0,
	\end{equation}
	where 
	\begin{equation*}
		\Phi_{d,i}=\begin{bmatrix}
			\tilde{\Xi}_{d,i}\\
			\Theta_{d}
		\end{bmatrix},~~ \tilde{S}_{i}=\begin{bmatrix}
			-\tilde{\Psi}_{i}\\
			I\\
			0
		\end{bmatrix},~~\tilde{W}_{i}=\begin{bmatrix}
			\tilde{\mathcal{D}}_{i}\\0	
		\end{bmatrix}
	\end{equation*}
	and $\tilde{\Xi}_{d,i}$, $\tilde{\Psi}_{i}$ and $\tilde{\mathcal{D}}_{i}$ are the rows of $\Xi_{d}$, $\Psi$ and $\mathcal{D}$ coupled with the considered set active constraints. By leveraging on \eqref{eq:KKT1}, we can now express our optimization variable $G_{i}(t)$ as a function of the Lagrange multipliers, \emph{i.e.,}
	\begin{equation}\label{eq:GasLambda}
		G_{i}(t)=-(\mathcal{H}_{d}^{\gamma})^{-1}\underbrace{\begin{bmatrix}
				\tilde{\Xi}_{d,i}^{\top} & \Theta_{d}^{\top}
		\end{bmatrix}}_{\Phi_{d,i}^{\top}}\tilde{\Lambda}_{i},
	\end{equation} 
	where 
	\begin{equation*}
		\tilde{\Lambda}_{i}=\begin{bmatrix}
			\tilde{\lambda}_{i}\\
			\mu
		\end{bmatrix}.
	\end{equation*} 
	We can now replace the latter into \eqref{eq:ActiveCond} to obtain an explicit expression of the Lagrange multipliers as functions of the matrices characterizing \eqref{eq:regmpQP_ddpc}:
	\begin{equation}\label{eq:explicitLambda}
		\tilde{\Lambda}_{i}=-\underbrace{\left[\Phi_{d,i}\left(\mathcal{H}_{d}^{\gamma}\right)^{-1}\Phi_{d,i}^{\top}\right]^{-1}}_{\Upsilon_{d,i}}\left(\tilde{S}_{i} x(t)\!+\!\tilde{W}_{i}\right).
	\end{equation} 
	In turn, this allows us to explicitly retrieve $G_{i}(t)$ as:
	\begin{equation}\label{eq:explicitG}
		G_{i}(t)\!=\!\left(\mathcal{H}_{d}^{\gamma}\right)^{-1}\!\Phi_{d,i}^{\top}\Upsilon_{d,i}\left(\tilde{S}_{i}x(t)\!+\!\tilde{W}_{i}\right),
	\end{equation}
	and the associated optimal input sequence as
	\begin{equation}\label{eq:optimal_seq}
		U_{i}(x(t))\!=\!\mathcal{V}\left(\mathcal{H}_{d}^{\gamma}\right)^{-1}\!\Phi_{d,i}^{\top}\Upsilon_{d,i}\left(\tilde{S}_{i}x(t)\!+\!\tilde{W}_{i}\right).
	\end{equation}
	Thus, the input to be fed to the system when the $i$-th set of constraints is active is defined as:
	\begin{equation}\label{eq:optimal_input}
		u_{i}(x(t))\!=\!\left[U_{i}(x(t))\right]_{1:m}\!.
	\end{equation}
	Through \eqref{eq:KKT3} and \eqref{eq:KKT4}, we can finally define the polyhedral region associated with the considered combination of active constraints, which is dictated by the following inequalities:
	\begin{subequations}\label{eq:polyregion}
		\begin{align}
			& \Upsilon_{d,i}\left(\tilde{S}_{i}x(t)\!+\!\tilde{W}_{i}\right)\leq 0,\\
			&\Xi_{d}\!\left(\mathcal{H}_{d}^{\gamma}\right)^{-1}\!\Phi_{d,i}^{\top}\Upsilon_{d,i}\!\left(\tilde{S}_{i}x(t)\!+\!\!\tilde{W}_{i}\right)\!\!+\!\Psi x(t)\!-\!\mathcal{D}\!\leq 0. 
		\end{align} 
	\end{subequations}
	The complete data-driven expression for \eqref{eq:param_input} is then straightforwardly obtained by following the above steps for all possible combinations of the active constraints. This operation ultimately yields an optimal input sequence $U(x(t))$ of the form:
	\begin{equation}\label{eq:optimalInput}
		U(x(t))=\begin{cases}
			U_{1}(x(t)), \mbox{ if } \mathcal{F}_{d,1}x(t) \leq \mathcal{E}_{d,1},\\
			\vdots\\
			U_{M}(x(t)), \mbox{ if } \mathcal{F}_{d,M}x(t) \leq \mathcal{E}_{d,M},
		\end{cases}
	\end{equation}
	where $M$ is given by the number of possible combinations of active constraints, $U_{i}$ corresponds to \eqref{eq:optimal_seq}, for all $i \in \{1,\ldots,M\}$, while $\{\mathcal{F}_{d,i},\mathcal{E}_{d,i}\}_{i=1}^{M}$ can be easily obtained from \eqref{eq:polyregion}. Consequently, the input to be fed to $\mathcal{S}$ starting from $x(t)$ can be retrieved by evaluating the PWA law
	\begin{equation}\label{eq:PWAlaw}
		u(x(t))=\begin{cases}
			u_{1}(x(t)), \mbox{ if } \mathcal{F}_{d,1}x(t) \leq \mathcal{E}_{d,1},\\
			\vdots \\
			u_{M}(x(t)), \mbox{ if } \mathcal{F}_{d,M}x(t) \leq \mathcal{E}_{d,M},\\
		\end{cases}
	\end{equation}
	with $u_{i}(x(t))$ given by \eqref{eq:optimal_input}, for $i=1,\ldots,M$.
	\begin{remark}[On Assumption \ref{ass:nondeg}]
		Although introduced to ease computations, we remark that Assumption~\ref{ass:non_degenrate} is not restrictive. Indeed, degenerate cases can be straightforwardly handled via existing approaches, \emph{e.g.,} see \cite{Bemporad2002b}.  		
	\end{remark}
	\begin{remark}[Data-driven and model-based]\label{remark:equivalence}
		Within a noiseless setting, the results in Theorem~\ref{thm:CL_predictor} and the one-to-one correspondence between the chosen parameterization of the control law in \eqref{eq:param_input} and its model-based counterpart guarantee the equivalence between \eqref{eq:MPC} and \eqref{eq:DDPC}. Therefore, when there is no noise, the data-driven explicit controller coincides with the E-MPC law as $\gamma \rightarrow 0$. 
	\end{remark}
	
	\subsection{Implementing Explicit DDPC}
	Based on the available batch of data and the features of the considered predictive control problem, the explicit DDPC law can be completely retrieved \textit{offline} from the available measurements, as summarized in Algorithm~\ref{algo1}.
	
	\begin{algorithm}[!tb]
		\caption{Offline construction of the explicit law} \label{algo1}
		~\textbf{Input}: Dataset $\mathcal{D}_{T}$; penalties $Q, P \succeq 0$; $R \succ 0$; horizon $N\!>\!0$; constraint matrices $C_{x},C_{u},d$; regularization parameter $\gamma>0$.
		\vspace*{.1cm}\hrule\vspace*{.1cm}
		\begin{enumerate}[label=\arabic*., ref=\theenumi{}]
			\item\label{step:1} \textbf{Construct} the data-based matrices $X_{1,L,T-L}$, $U_{0,L,T-L-1}$, $X_{0,T-L-1}$.
			\item\label{step:2} \textbf{Build} $\mathcal{H}_{d}^{\gamma}$, $\Xi_{d}$, $\Psi$, $\Theta_{d}$ in \eqref{eq:cost_matrix} based on the cost and constraints of the DDPC problem.
			\item\label{step:3} \textbf{Find} all possible combinations of active constraints. 
			\item\label{step:4} \textbf{For each} combination, \textbf{isolate} the matrices $\tilde{\Xi}_{d}$, $\tilde{W}$ and $\tilde{S}$ characterizing \eqref{eq:ActiveCond}
			\item\label{step:5} \textbf{If not} all rows of $\tilde{\Xi}_{d}$ are \textbf{linearly independent}, \textbf{handle} the degeneracy, \emph{e.g.,} as in \cite{Bemporad2002b}.		
			\item\label{step:6} \textbf{Find} the PWA explicit law by retrieving \eqref{eq:optimal_seq}-\eqref{eq:polyregion} \textbf{for all} possible combinations of active constraints.
			\item\label{step:7} \textbf{Merge} polyhedral regions as in \cite{Bemporad2001}.
			\item\label{step:8} \textbf{Extract} the first component of the optimal input sequence $U(x(t))$.
		\end{enumerate}
		\vspace*{.1cm}\hrule\vspace*{.1cm}
		~\textbf{Output}: Optimal input $u(x(t))$.
	\end{algorithm} 
	
	Given the data, one has to initially construct the Hankel matrices needed to build the DDPC problem (see steps~\ref{step:1}-\ref{step:2}). Once all the possible combinations of active constraints have been detected at step~\ref{step:3} and degenerate scenarios have been handled (see step~\ref{step:5}), at step~\ref{step:6} the local controllers and the associated polyhedral regions are retrieved according to \eqref{eq:optimal_seq}-\eqref{eq:polyregion}. Lastly, at step~\ref{step:7}, the optimal control sequence is simplified, by merging polyhedral regions whenever possible. After this step, the explicit optimal input can simply be retrieved by extracting the first element of the input sequence $U(x(t))$ (see step~\ref{step:8}).
	
	Once the explicit DDPC law has been retrieved, the computation of the optimal input at each time instant simply consists of a function evaluation. Specifically, one has to $(i)$ search for the polyhedral region the current state $x(t)$ belongs to, and $(ii)$ apply the corresponding parametric law. We stress that this computational advantage is retained for simple control problems only (\emph{i.e.,} for short prediction horizon and small systems). Indeed, the complexity of the PWA law is known to rapidly increase \cite{Borrelli} with the one of the DDPC problem to be solved, analogously to the model-based case.   
	
	\subsection{Explicit data-driven predictive control and closed-loop stability}
	When designing a controller in a data-driven setting, it is crucial to check the stability of the resulting closed-loop system before the controller deployment. Towards this objective, we now show how the peculiar features of the explicit data-driven predictive control can be leveraged in combination with existing techniques to devise an off-line, data-driven stability test. To this end, let us assume that the $i$-th set of constraints is active and consider the following multi-step ahead closed-loop model:
	\begin{equation}\label{eq:multistep_cl}
		\hat{X}(t)\!=\!\mathcal{X}(\mathcal{H}_{d}^{\gamma})^{-1}\Phi_{d,i}^{\top}\Upsilon_{d,i}\tilde{S}_{i}x(t)\!+\!\mathcal{X}(\mathcal{H}_{d}^{\gamma})^{-1}\Phi_{d,i}^{\top}\Upsilon_{d,i}\tilde{W}_{i},
	\end{equation}
	obtained by combining \eqref{eq:CLpred} with the result of our explicit derivation in \eqref{eq:explicitG}, where $\hat{X}$ stacks the state predicted by the learned closed-loop model, \emph{i.e.,}
	\begin{equation*}
		\hat{X}(t)=		\begin{bmatrix}
			\hat{x}(t+1)\\
			\vdots\\
			\hat{x}(t+L)
		\end{bmatrix}.
	\end{equation*}
	From \eqref{eq:multistep_cl}, we can then isolate the learned one-step ahead closed-loop model, namely
	\begin{subequations}
		\begin{equation}
			\hat{x}(t+1)=A_{d,i}^{cl}x(t)+f_{d,i}^{cl},
		\end{equation}
		where
		\begin{align}
			&A_{d,i}^{cl}=[\mathcal{X}(\mathcal{H}_{d}^{\gamma})^{-1}\Phi_{d,i}^{\top}\Upsilon_{d,i}\tilde{S}_{i}]_{1:n,1:n},\\
			&f_{d,i}^{cl}=[\mathcal{X}(\mathcal{H}_{d}^{\gamma})^{-1}\Phi_{d,i}^{\top}\Upsilon_{d,i}\tilde{W}_{i}]_{1:n}.
		\end{align} 
	\end{subequations}
	When performed for each mode $i \in \{1,\ldots,M\}$, these manipulations allow us to retrieve the data-driven closed-loop transition matrix $A_{d,i}^{cl}$ for each $i \in \{1,\ldots,M\}$. Retrieving these matrices ultimately enables us to apply model-based techniques, \emph{e.g.,} the ones presented in \cite{Mignone}, to shed a light on the features of the final closed-loop. As an example, one can search for a matrix $\mathcal{P} \in \mathbb{R}^{n \times n}$ satisfying the following sufficient conditions for asymptotic stability:
	\begin{subequations}\label{eq:test1}
		\begin{align}
			& \mathcal{P}>0\\
			& (A_{d,i}^{cl})^{\top}\mathcal{P}A_{d,i}^{cl}-P <0, ~~i=1,\ldots,M,
		\end{align}
	\end{subequations}
	or, alternatively, look for the set of matrices $\{\mathcal{P}_{i} \in \mathbb{R}^{n\times n}\}_{i=1}^{M}$ verifying the following LMIs:
	\begin{subequations}\label{eq:test2}
		\begin{align}
			& \mathcal{P}_{i}>0,~~i=1,\ldots,M,\\
			&(A_{d,i}^{cl})^{\top}\mathcal{P}_{j}A_{d,i}^{cl}-P <0, ~~i,j=1,\ldots,M, 
		\end{align}
		which are also sufficient conditions for asymptotic stability.
	\end{subequations}
	\begin{remark}[On the tuning of $P$ in \eqref{eq:MPCcost}]
		When $\mathcal{S}$ is known to be open-loop stable, the terminal weight $P \succeq 0$ in \eqref{eq:MPCcost} is generally selected as the solution of the Lyapunov equation
		\begin{equation*}
			P=A^{\top}PA+Q.
		\end{equation*} 	
		This equation can be directly translated into its data-driven counterpart by exploiting \cite{DePersis2019} as follows:
		\begin{equation}\label{eq:DDLyap}
			P=A_{d}PA_{d}+Q,
		\end{equation}
		with
		\begin{equation*}
			A_{d}=X_{1,T}\!\begin{bmatrix}
				X_{0,T-1}\\
				\hline
				U_{0,1,T-1}
			\end{bmatrix}^{\dagger}\!\!\begin{bmatrix}
				I\\
				0
			\end{bmatrix},
		\end{equation*}
		thus providing a data-based approach for the selection of this parameter. 
	\end{remark}
	\begin{remark}[Hyper-parameter tuning]
		The possibility of performing an off-line data-based stability check on the data-driven explicit law can be useful to preliminarily assess the effects of different choices of the tuning parameters $Q$, $R$ and $P$ in \eqref{eq:MPCcost} and $\gamma$ in \eqref{eq:reg_cost}, allowing one to discard the ones resulting in a failure of the data-driven stability tests.
	\end{remark}
	
	\subsection{Regularization and noise handling}\label{Sec:noise}
	As highlighted in Remark~\ref{remark:noisefree}, all the equivalences we rely on to derive the explicit predictive control law are verified when $\mathcal{D}_{T}$ is noiseless. However, in practice, $\Omega$ in Assumption~\ref{ass:noisy_outputs} is generally a non-zero matrix. 
	
	To cope with noisy data, we follow the footsteps of \cite{Berberich2021,Dorfler2021} and propose to leverage on the regularization term introduced in \eqref{eq:reg_cost}. Indeed, as in standard ridge regression \cite{hastie2009elements}, this additional element of the cost steers all the components of $G(t)$ towards small values. In turn, this potentially limits the impact of noise on the constraints in \eqref{eq:regmpQPconstr}-\eqref{eq:regmpQPbuild} and, thus, on the final explicit law. This shrinkage effect is modulated by the regularization parameter $\gamma$, with the reduction in the magnitude of $G(t)$ being stronger whenever large values of $\gamma$ are considered. At the same time, $\gamma$ implicitly changes the balance between the penalties in the original data-driven control problem in \eqref{eq:DDPC}, with excessively high values of $\gamma$ potentially driving the explicit data-driven controller far away from its model-based counterpart. The choice of this hyper-parameters thus becomes an important \textit{tuning-knob} of the approach, requiring one to trade-off between handling noise and keeping explicit DDPC as close as possible to the implicit DDPC problem. 
	
	Although the stability checks in \eqref{eq:test1}-\eqref{eq:test2} can be used to have a preliminary assessment on the effect of different choices of $\gamma$, at the moment this balance can only be attained through closed-loop trials for several values of $\gamma$. Such a procedure allows one to ultimately select the hyper-parameter that best fits one's needs, at the price of requiring closed-loop experiments that can be rather safety-critical in practice, especially when a simulator of the system is not available. 
	
	Whenever multiple experiments can be performed by feeding the plant with the same input sequence $\mathcal{U}_{T}$, the burden associated to the choice of $\gamma$ can be alleviated by exploiting the features of Assumption~\ref{ass:noisy_outputs} itself. In this scenario, one can indeed replace $\mathcal{D}_{T}$ with the averaged dataset $\bar{\mathcal{D}}_{T}=\{\mathcal{U}_{T},\bar{\mathcal{Y}}_{T}\}$, where $\bar{\mathcal{Y}}_{T}=\{\bar{y}_{t}\}_{t=0}^{T}$ and 
	\begin{equation}\label{eq:averaged_y}
		\bar{y}_{t}=\frac{1}{N}\sum_{i=1}^{N}y_{t}^{(i)},
	\end{equation}    
	with $y_{t}^{(i)}$ denoting the output of the $i$-th experiment. Since the noise is assumed to be zero mean, the law of large numbers asymptotically yields
	\begin{equation}\label{eq:asymp_convergence}
		\bar{y}_{t} \underset{N \rightarrow \infty}{\longrightarrow} y_{t}^{\mathrm{o}}.
	\end{equation} 
	As such, when the number $N$ of experiments increases, the role of $\gamma$ in handling noise is progressively less dominant. In this case, $\gamma$ should then be used \textit{only} to make the DDPC problem well-defined. Any small $\gamma>0$ is acceptable for this purpose.
	
	\section{Numerical examples}\label{sec:examples}
	The performance of the explicit predictive controller are now assessed on two benchmark examples: $(i)$ the regulation to zero of the stable open-loop system of \cite{Bemporad2002b}, for the case when the state is fully measurable; 
	and $(ii)$ the altitude control of a quadcopter. Since the last  example features an open-loop unstable linearized plant, data are collected in closed-loop, by assuming that the drone is stabilized by an (unknown) pre-existing controller. In both the examples, the level of noise acting on the measured states is assessed through the averaged Signal-to-Noise-Ratio (SNR):
	\begin{equation}\label{eq:SNR}
		\overline{\mbox{SNR}}\!=\!\frac{1}{n}\!\sum_{i=1}^{n} 10 \log_{10}\!\left(\!\frac{\sum_{t=0}^{T}(x_{i}(t)-w_{i}(t))^{2}}{\sum_{t=0}^{T} (w_{i}(t))^{2}}\!\right)\!,~\mbox{[dB]}
	\end{equation}
	where $x_{i}(t)$ and $w_{i}(t)$ denote the $i$-th components of the state and the measurement noise, respectively. All computations are carried out on an Intel Core i7-7700HQ processor, running MATLAB 2019b.    
	
	\subsection{Open-loop stable benchmark system}
	\begin{figure}[!tb]
		\centering
		\includegraphics[scale=.375]{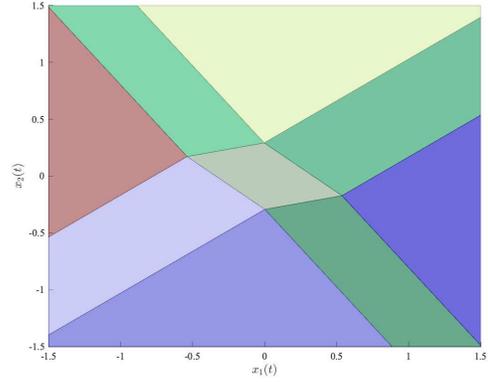}
		\caption{Open-loop stable benchmark: polyhedral partition of the explicit data-driven law.}\label{fig:OL_partition}
	\end{figure}
	\begin{figure}[!tb]
		\centering
		\begin{tabular}{c}
			\subfigure[State trajectories.]{\includegraphics[scale=0.35]{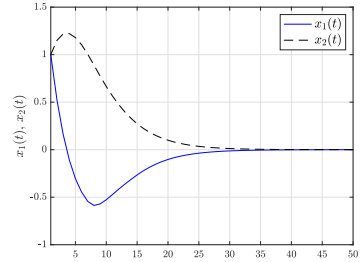}}\\ \subfigure[Input]{\includegraphics[scale=0.35]{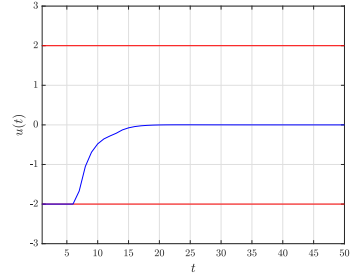}}
		\end{tabular}
		\caption{Open-loop stable benchmark: state and input trajectories obtained with the explicit data-driven law.}\label{fig:OL}
	\end{figure}
	Let us consider the benchmark system described by:
	\begin{equation}
		\mathcal{S}:~x(t\!+\!1)\!=\!\!\begin{bmatrix}
			0.7326 & -0.0861\\
			0.1722 & 0.9909
		\end{bmatrix}x(t)\!+\!\begin{bmatrix}
			0.0609\\
			0.0064
		\end{bmatrix}u(t).
	\end{equation}
	Our goal is to regulate both components of the state to zero, while enforcing the following box-constraint on the input:
	\begin{equation}
		-2 \leq u(k) \leq 2.
	\end{equation}
	Towards this goal, we collect a set $\mathcal{D}_{T}$ of $T=100$ input/state pairs, by feeding $\mathcal{S}$ with an input sequence uniformly distributed within the interval $[-5,5]$. According to Assumption~\ref{ass:noisy_outputs}, the measured states are corrupted by an additive zero-mean white noise sequence, with variance yielding $\overline{\mbox{SNR}}=20$~dB. The parameters characterizing the DDPC problem to be solved are selected as in \cite{Bemporad2002b}, namely $L=2$, $Q=I$, $R=0.01$ and $P$ is chosen as the solution of the data-driven Lyapunov in \eqref{eq:DDLyap}. By setting $\gamma=10$, the partition associated with the explicit data-driven predictive controller\footnote{The partition is plotted thanks to the Hybrid Toolbox~\cite{HybTBX}.} is the one shown in \figurename{~\ref{fig:OL_partition}}, which approximately correspond to that reported in \cite{Bemporad2002b}\footnote{The negligible differences with respect to the model-based partition are due to the noise on the batch data.}. Prior to the controller deployment, we have performed the data-based closed-loop stability check in \eqref{eq:test1}, resulting in\footnote{The LMIs in \eqref{eq:test1} are solved with CVX \cite{cvx,gb08}.} 
	\begin{equation*}
		\mathcal{P}=\begin{bmatrix}
			24.8695 & 10.5595\\
			10.5595 & 43.2657
		\end{bmatrix} \succ 0.
	\end{equation*}   
	This indicates that the explicit law preserves the stability of the open-loop system. \figurename{~\ref{fig:OL}} report the trajectories of the state and the optimal input obtained over a noise-free closed-loop tests with the explicit data-driven law, which confirm its effectiveness and validate the result of the data-driven stability check.
	\begin{figure*}[!tb]
		\centering
		\begin{tabular}{ccc}
			\subfigure[$\overline{\mbox{SNR}}=$40 dB]{\includegraphics[scale=0.3,,trim=.75cm 0cm 0cm 0cm,clip]{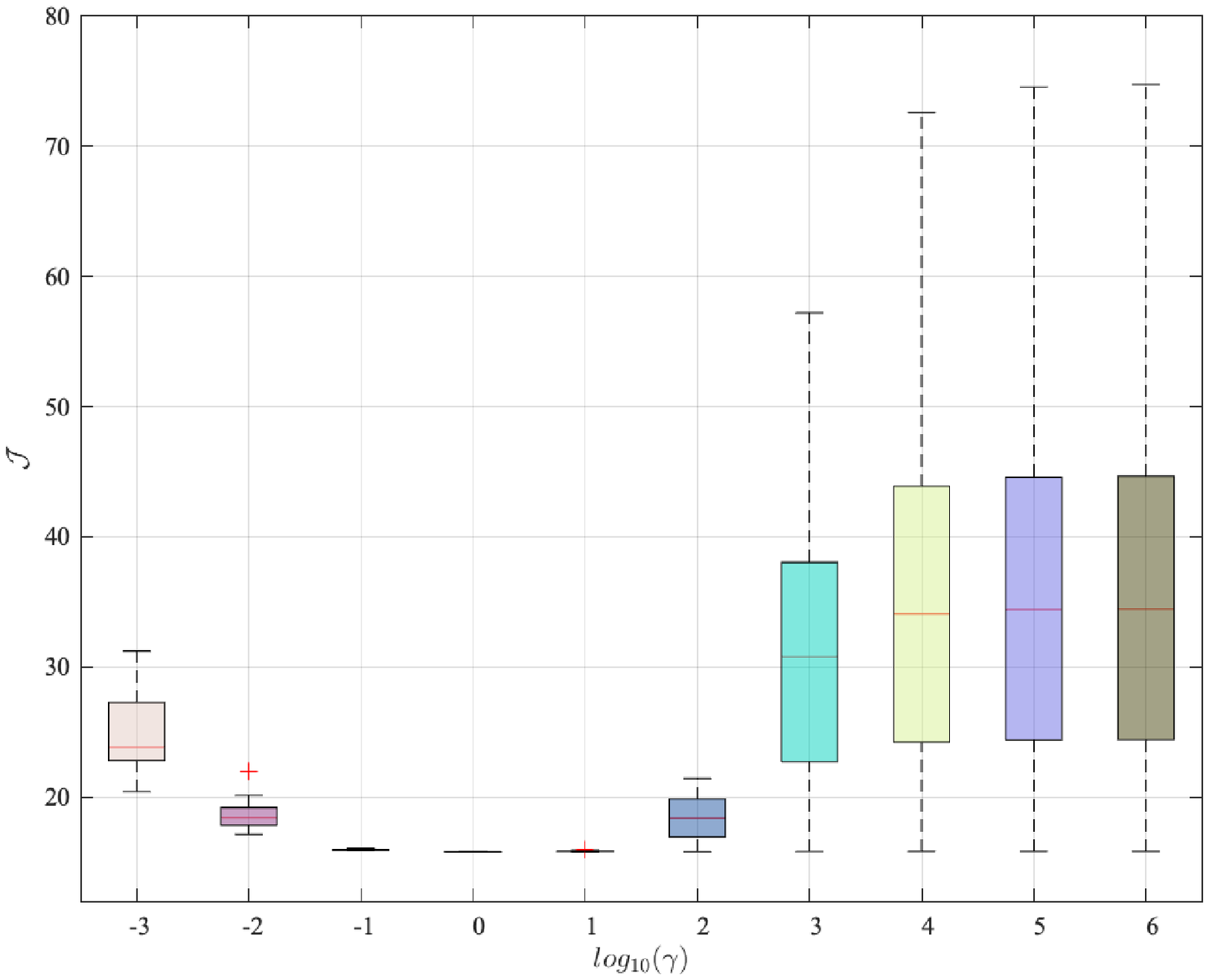}}&
			\subfigure[$\overline{\mbox{SNR}}=$20 dB]{\includegraphics[scale=0.3,,trim=.75cm 0cm 0cm 0cm,clip]{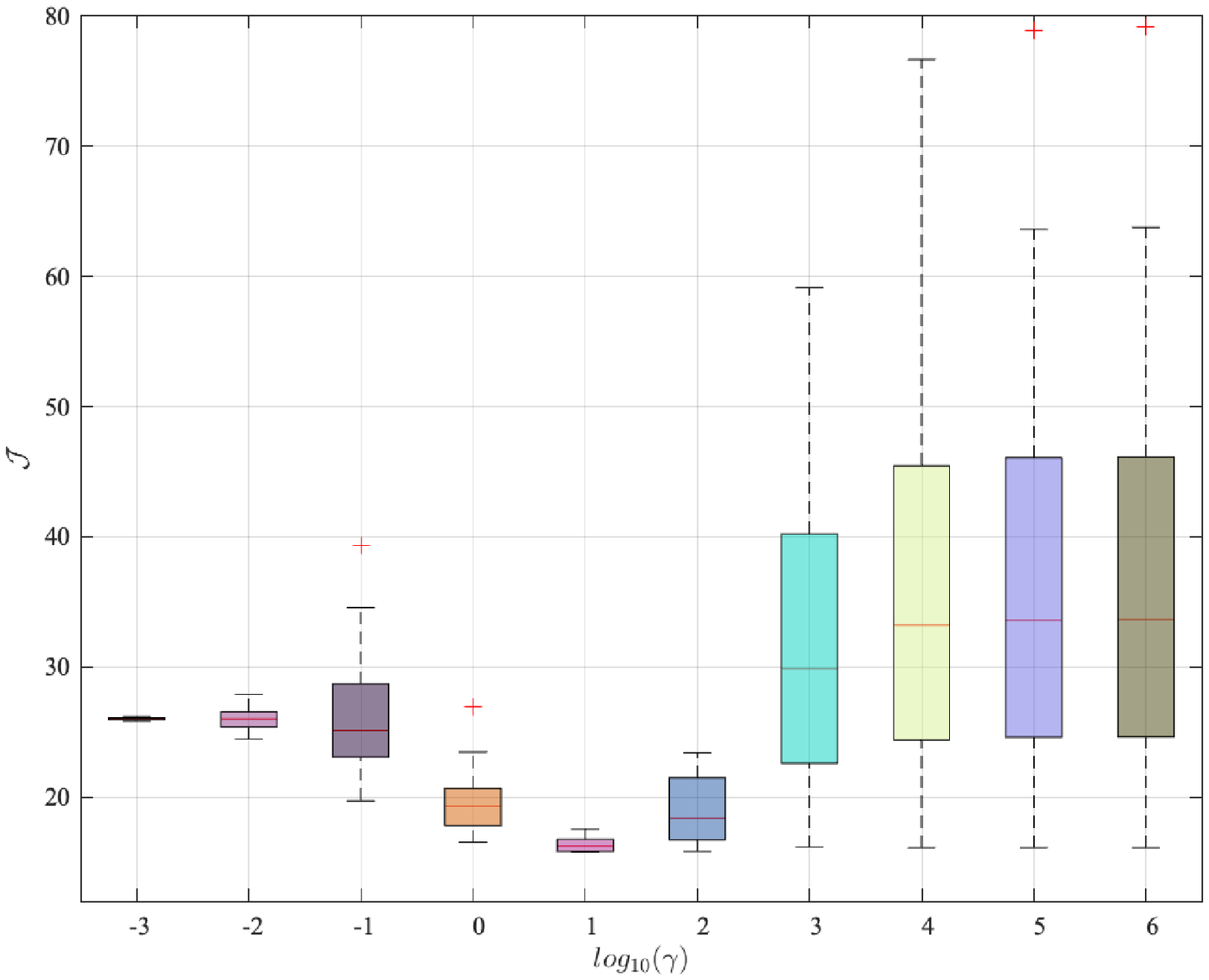}} &
			\subfigure[$\overline{\mbox{SNR}}=$10 dB]{\includegraphics[scale=0.3,,trim=.75cm 0cm 0cm 0cm,clip]{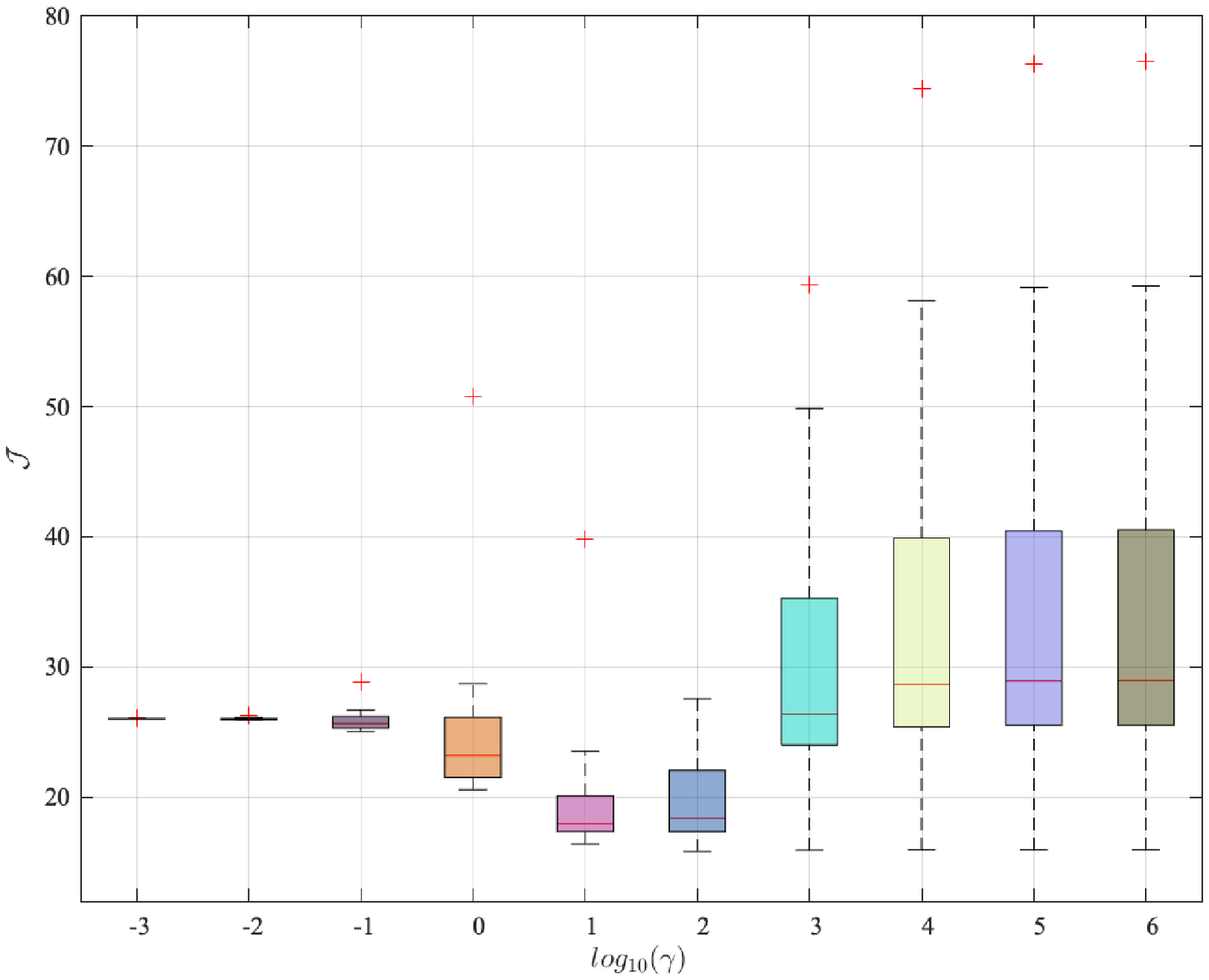}}
		\end{tabular}
		\caption{Open-loop stable benchmark: $\mathcal{J}$ in \eqref{eq:performance_index} \emph{vs} $\log_{10}(\gamma)$ for and increasing levels of noise over 30 realizations of $\mathcal{D}_{T}$.}\label{fig:BOXPLOT}
	\end{figure*}
	
	We now assess the sensitivity of the explicit controller to the choice of $\gamma$ over $20$ Monte-Carlo realizations of the batch datasets $\mathcal{D}_{T}$, for different noise levels. This evaluation is performed by looking at the cost of the controller over the same noiseless closed-loop test of length $T_{v}=50$ considered previously, \emph{i.e.,}
	\begin{equation}\label{eq:performance_index}
		\mathcal{J}=\sum_{t=0}^{T_{v}} \|x(t)\|_{Q}^{2}+\|u(t)\|_{R}^{2}.
	\end{equation}      
	As shown in \figurename{~\ref{fig:BOXPLOT}}, the value of $\gamma$ that leads to the minimum closed-loop cost tends to decrease when the noise level increases, supporting our considerations in Section~\ref{Sec:noise}. At the same time, by properly selecting $\gamma$ we attain a cost $\mathcal{J}=16.37 \pm 0.58$, which is generally close to the oracle $\mathcal{J}^{\mathcal{O}}=15.77$,
	which is the one achieved by the \emph{oracle} law, \emph{i.e.,} the model-based predictive controller obtained by exploiting the true model of $\mathcal{S}$. These results additionally show that, for increasing levels of noise, the choice of $\gamma$ becomes more challenging, since the range of values leading to the minimum $\mathcal{J}$ progressively shrinks. Note that, when $\gamma$ is excessively small the optimal input is always zero and $\mathcal{S}$ evolves freely\footnote{This behavior is also observed for $\gamma\!\leq\!10^{-6}$, $\gamma\!\leq\!10^{-4} $ and $\gamma\!\leq\!10^{-3}$ when $\overline{\mbox{SNR}}\!=\!40$~dB, $\overline{\mbox{SNR}}\!=\!20$~dB and $\overline{\mbox{SNR}}\!=\!10$~dB, respectively.}.     

	\begin{figure*}[!tb]
		\centering
		\begin{tabular}{ccc}
			\subfigure[$N=1$]{\includegraphics[scale=0.3,,trim=.75cm 0cm 0cm 0cm,clip]{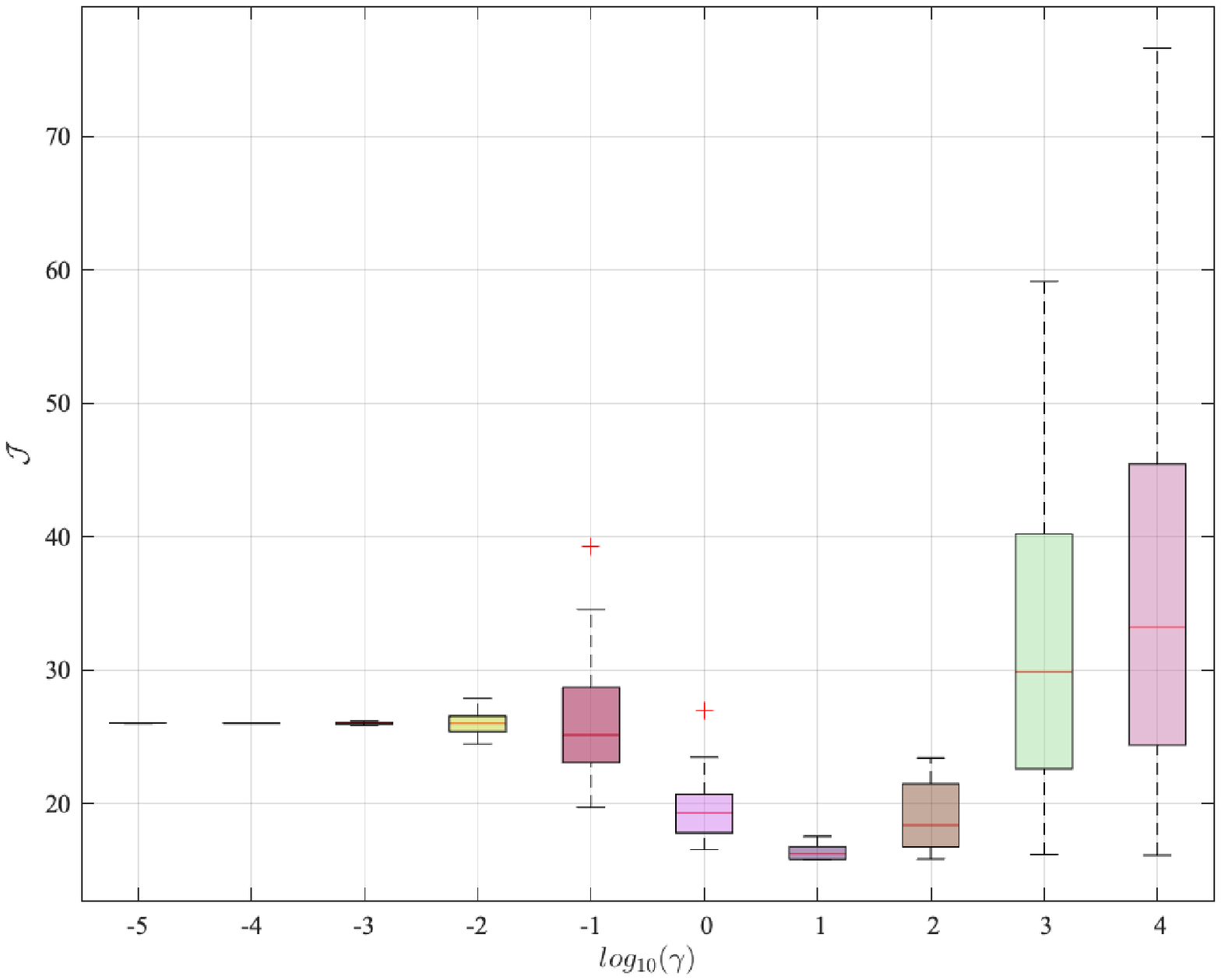}}&
			\subfigure[$N=10$]{\includegraphics[scale=0.3,trim=.75cm 0cm 0cm 0cm,clip]{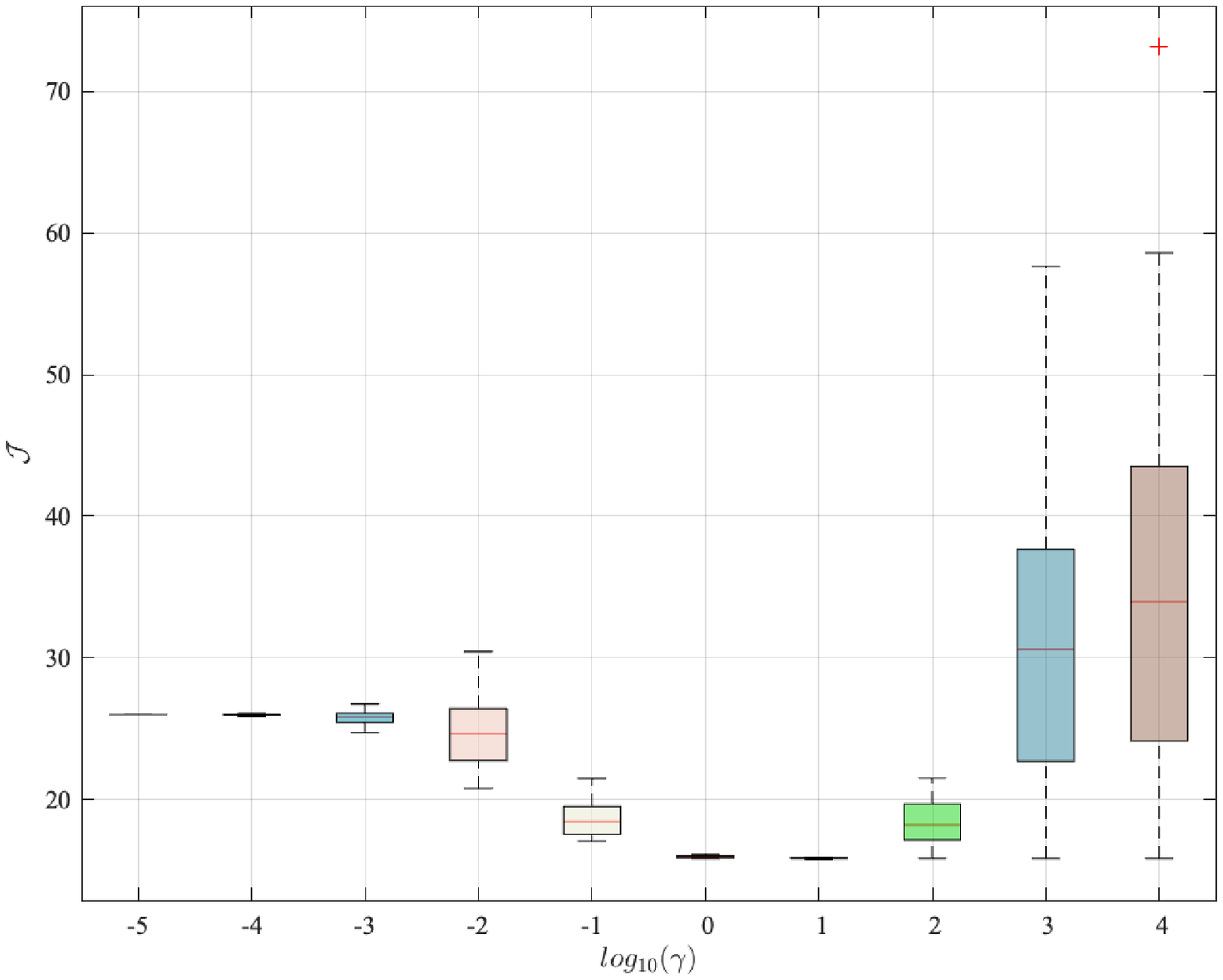}} &
			\subfigure[$N=100$]{\includegraphics[scale=0.3,trim=.75cm 0cm 0cm 0cm,clip]{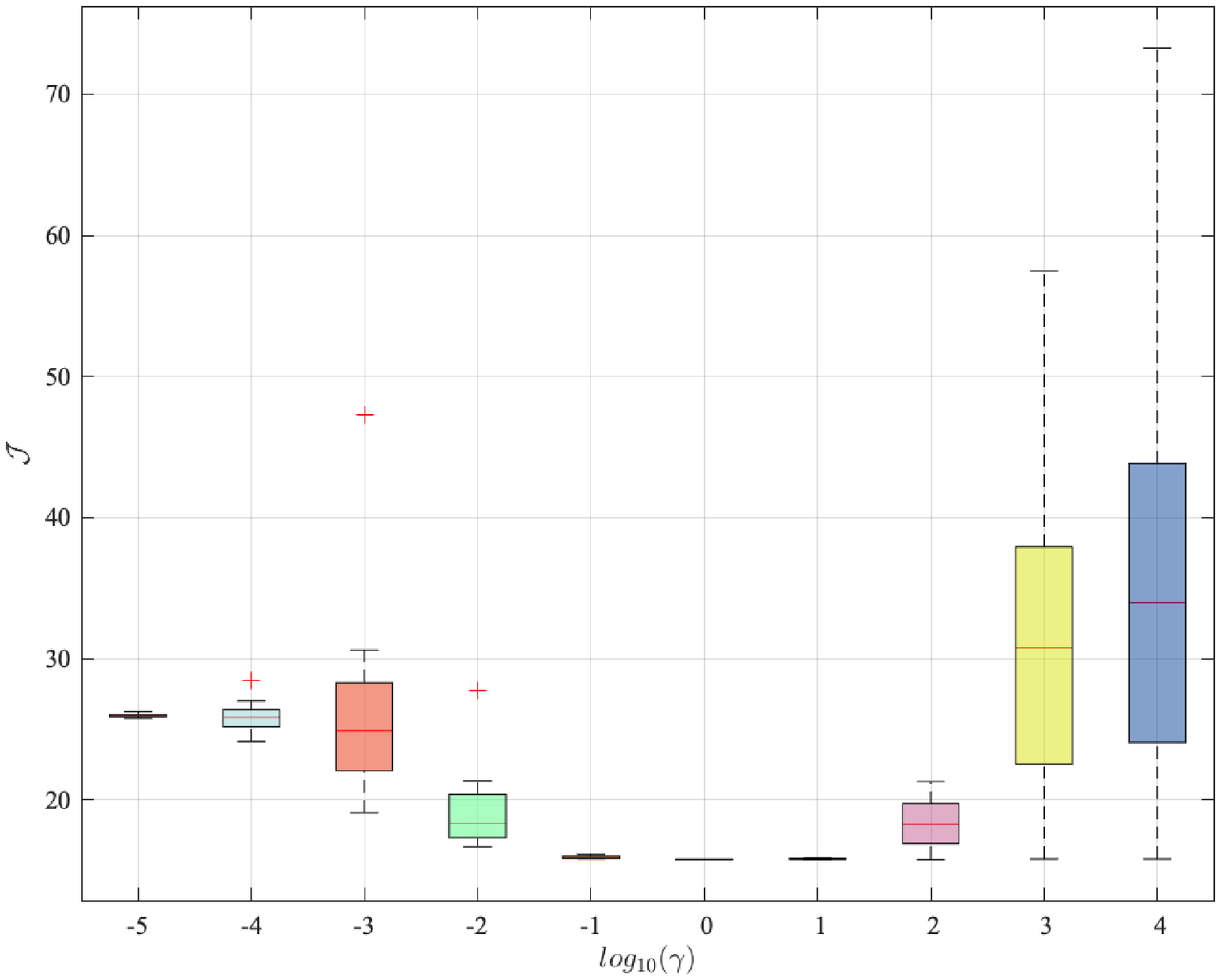}}
		\end{tabular}
		\caption{Open-loop stable benchmark: $\mathcal{J}$ in \eqref{eq:performance_index} \emph{vs} $\log_{10}(\gamma)$ for an increasing number of repeated experiments over 30 realizations of $\mathcal{D}_{T}$.}\label{fig:BOXPLOT2}
	\end{figure*}
	
	We additionally evaluate the effect of averaging, by looking at the performance index $\mathcal{J}$ in \eqref{eq:performance_index} over $30$ Monte-Carlo data-collections for an increasing number $N$ of repeated experiments of length $T=100$. The measurements are affected by noise, yielding $\overline{\mbox{SNR}}=20$~dB. \figurename{~\ref{fig:BOXPLOT2}} shows that the use of the averaged dataset $\bar{\mathcal{D}}_{T}$ has a similar effect to a reduction of the noise level. Indeed, for increasing $N$ the optimal $\gamma$ slowly shifts towards smaller values, thus further implying the gradual reduction in the impact of $\gamma$ on noise handling.

	\subsection{Altitude control of a quadcopter}
	\begin{figure}[!tb]
		\centering
		\begin{tabular}{c}
			\subfigure[Take-off]{\includegraphics[scale=.6]{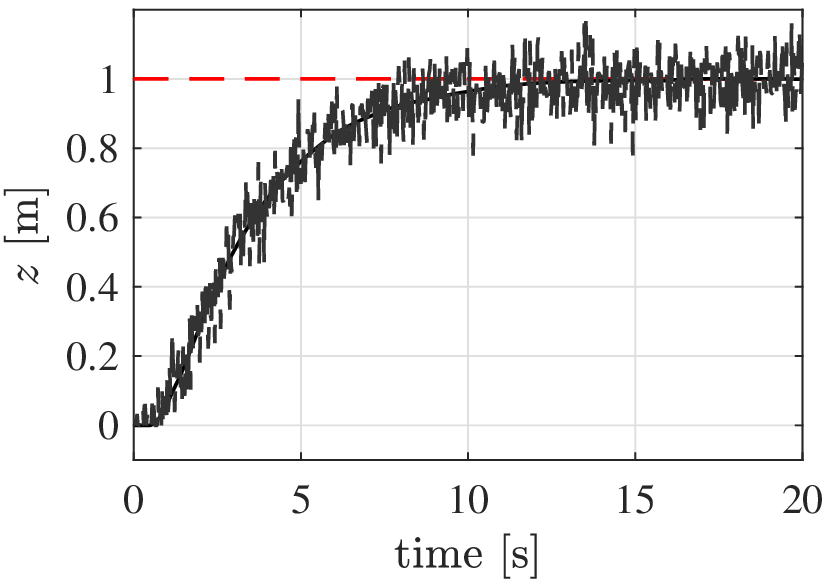}}\\
			\subfigure[Landing]{\includegraphics[scale=.6]{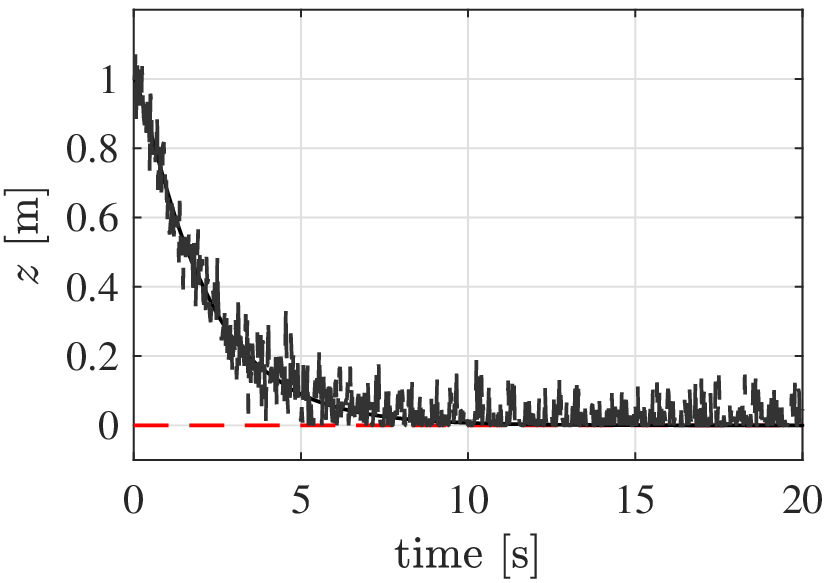}}
		\end{tabular}
		\caption{Altitude control: measured (dotted-dashed gray line) and actual (black line) altitude \emph{vs} reference (dashed red line).}\label{fig:noisy_testQUAD}
	\end{figure}
	As a final case study, we consider a nonlinear system, namely the problem of controlling the altitude of a quadcopter, to perform landing or take-off maneuvers. To this end, we exploit the same simulator used in \cite{Formentin2011} to collect the data and to carry out the closed-loop experiments with the learned explicit law. Let $z(t)$~[m] be the altitude of the quadcopter, $v_{z}(t)$~[m/s] be its vertical velocity and $(\theta(t),\phi(t),\psi(t))$~[deg] its roll, pitch and yaw angles at time $t$. Both the the altitude $z(t)$~[m] and the vertical velocity $v_{z}(t)$ are assumed to be measured, with the measurement being corrupted by a zero-mean white noise, resulting in $\overline{\mbox{SNR}} =$ 30~dB over these two outputs. As this system is open-loop unstable, the data collection phase is carried out in closed-loop for $20$~[s] at a sampling rate of $40$~[Hz], by using the four \emph{proportional derivative} (PD) controllers introduced in \cite{Formentin2011}. The altitude set point used at this stage is generated uniformly at random in the interval $[0,4]$~[m]. The set points for all the attitude angles are instead selected as slowly variable random signals within the interval $[-0.2,0.2]$~[rad]. These choices yield a dataset $\mathcal{D}_{T}$ of length $T=800$, that satisfies Assumption~\ref{ass:persistency of excitation} and allows us to retain information on possible non-zero angular configurations.  
	
	The three attitude controllers introduced in \cite{Formentin2011} are further retained in testing to keep the attitude angles at zero and to decouple the altitude dynamics from that of the other state variables. Within this setting, the explicit data-driven law is designed by imposing $L=5$, $Q=P=\mbox{diag}([1,0.1])$, $R=10^{-5}$ and $\gamma=1$. To mitigate the effect of the gravitational force, the design and closed-loop deployment of the designed explicit controller are carried out by pre-compensating it. As a result, the input to be optimized is
	\begin{equation}\label{eq:compensation}
		u(t)=\frac{u_{z}(t)}{m}-g,
	\end{equation}
	where $m=0.5$~[kg] is the mass of the quadcopter, $g=9.81$~[m/s] is the gravitational acceleration and $u_{z}(t)$ is the input prior to the compensation. A similar approach is adopted for the control problem to fit our framework in both landing and take-off scenarios. We thus consider the reduced state
	\begin{equation}
		x(t)=\begin{bmatrix}
			z(t)-\bar{z}\\
			v_{z}(t)
		\end{bmatrix},
	\end{equation}
	where $\bar{z}$~[m] is the altitude set point. To avoid potential crashes of the quadcopter, in designing the explicit law we impose the following constraint on the state of the system:
	\begin{equation}
		x_{1}(t) \geq -\bar{z},
	\end{equation}
	which, in turn, guarantees the altitude to be always non-negative. Meanwhile, the pre-compensated input is constrained to the interval:
	\begin{equation}
		-9.81 \leq u(t) \leq 9.564,
	\end{equation} 
	where the lower bound corresponds to a null input and the upper limit is dictated by the maximum power of the motors\footnote{The reader is referred to \cite{Formentin2011} for additional details on the system.}.
	
	The performance of the learned explicit law attained in take-off and landing are reported in \figurename{~\ref{fig:noisy_testQUAD}}. Here we consider closed-loop tests in which the altitude and the vertical velocity are noisy, with the noise acting on the closed-loop measurements sharing the features of that corrupting the batch ones. Despite the noise acting on the initial condition at each step, both maneuvers are successfully performed, thus showing the effectiveness of the retrieved explicit data-driven laws.  
	
	\section{Conclusions}\label{sec:conclusions}
	By leveraging on the known PWA nature of the explicit MPC law within linear quadratic predictive control, in this paper we propose an approach to derive such an explicit controller from data only, without undertaking a full modeling/identification step. Thanks to the formalization of the problem, well-known model-based techniques can be straightforwardly adapted to check the stability of the closed-loop system before deploying the controller.\\
	Future research will be devoted to extend these preliminary results to cases in which the state is not fully measurable, to exploit priors to guarantee practical closed-loop stability by design. Future work will also be devoted to formalize the connections between the explicit solution proposed in this paper and the one introduced in \cite{Breschi2021b}, consequently providing a comparative analysis of the two approaches.
	
	\bibliographystyle{plain}
	\bibliography{EDDPC} 
	
\end{document}